\documentclass{conm-p-l}

\copyrightinfo{2015}{}

\usepackage{graphicx}
\usepackage{subfigure}

\usepackage{amssymb,amsmath,amsthm,amscd}

\newtheorem{theorem}{Theorem}[section]

\theoremstyle{definition}

\theoremstyle{remark}

\numberwithin{equation}{section}

\newcommand{\ga}{\gamma}

\newcommand{\dl}{\delta}
\newcommand{\Dl}{\Delta}

\newcommand{\ra}{\rightarrow}

\newcommand{\sg}{\sigma}

\newcommand{\pa}{\partial}

\newcommand{\om}{\omega}
\newcommand{\Om}{\Omega}
\newcommand{\na}{\nabla}

\newcommand{\non}{\nonumber}

\begin{document}

\title[Explicit solutions]{Rough dependence upon initial data exemplified by explicit solutions and the effect of viscosity}

\author{Y. Charles Li}
\address{Department of Mathematics, University of Missouri, 
Columbia, MO 65211, USA}
\email{liyan@missouri.edu}
\urladdr{http://faculty.missouri.edu/~liyan}

\curraddr{}
\thanks{}

\subjclass{Primary 76, 35; Secondary 34}

\date{}

\dedicatory{}


\begin{abstract}
In this article, we present some explicit solutions showing rough dependence upon initial 
data. We also studies viscous effects. An extension of a theorem of Cauchy to the viscous case is also presented.
\end{abstract}

\maketitle

\tableofcontents

\section{Introduction}
It is well-known that both Navier-Stokes and Euler equations are locally well-posed in the Sobolev space $H^s$ ($s > \frac{d}{2} +1$) where $d$ is the spatial dimension \cite{Kat72} \cite{Kat75}. In two dimensions ($d=2$), the well-posedness is also global. The solutions are in the space $C^0([0, T), H^3)$, where $T$ is either finite or infinite. The solution operator $F^t$ is a one-parameter family of maps on $H^3$. For each fixed $t$, $F^t$ maps the initial condition $u(0)$ to the solution's value at time $t$, $u(t)$. $F^t$ is continuous in $u(0)$. In \cite{HM10}, a simple example was presented showing that the solution operator $F^t$ is not uniformly continuous in $u(0)$ under the Euler dynamics. In \cite{Inc12} \cite{Inc15}, the nowhere differentiability of $F^t$ under the Euler dynamics was proved. In \cite{Li14}, we derived the upper bound on the derivative of the solution operator under the Navier-Stokes dynamics: $\| \na F^t \| \leq e^{\sg \sqrt{Re} \sqrt{t} +\sg_1 t}$ where $\sg$ and $\sg_1$ are constants depending only on the norm of the base solution. In this article, we present explicit solutions to both 2D Euler and 2D Navier-Stokes equations. These solutions show that under the Euler dynamics, $\| \na F^t \| = \infty$, and under the  Navier-Stokes dynamics, $\| \na F^t \| \ra \infty$ as the Reynolds number approaches infinity. We name the above nature of dependence of the solution operator on initial data as rough dependence on initial data.

\section{Explicit solutions with rough dependence on initial data}

In this section, we present explicit solutions to the 2D Euler equations and 2D Navier-Stokes equations, which show rough dependence on initial data. Consider the following 2D Euler equations
\begin{equation}
\pa_t u + u \cdot \na u = - \na p , \  \na \cdot u = 0 , \label{2DE}
\end{equation}
and 2D Navier-Stokes equations
\begin{equation}
\pa_t u + u \cdot \na u = - \na p + \frac{1}{Re} \Dl u, \  \na \cdot u = 0 , \label{2DNS}
\end{equation}
under periodic boundary condition with period domain [$0, 2\pi$] $\times$ [$0, 2\pi$], where $u=(u_1,u_2)$ is the velocity, $p$ is pressure, and the spatial coordinate is denoted by $x=(x_1,x_2)$.
First we study the following simple explicit solution to the 2D Euler equations,
\begin{equation}
u_1 = \sum_{n=1}^\infty \frac{1}{n^{3+\ga }} \sin [n (x_2-\sg t)] , \  u_2 = \sg ,  \label{2Dsim}
\end{equation}
where $\frac{1}{2} < \ga \leq 1$, and $\sg$ is a real parameter. We view this solution as a solution in the space $C^0([0,\infty ) , H^3)$, where $H^3$ is the Sobolev space. We select the Sobolev space $H^3$ because both the 2D Euler equations and the 2D Navier-Stokes equations are well-posed in $H^3$. The same construction 
in this article can be easily extended to any Sobolev space $H^s$ ($s \in \mathbb{R}$). 
Denote by $F^t$ the solution operator of either the 2D Euler equations or the 2D Navier-Stokes equations. For any fixed $t$, one can view $F^t$ as a map in $H^3$, 
$$F^t \ : \ u(0) \ra u(t) . $$ 
The initial condition of the solution (\ref{2Dsim}) is 
\[
u_1 = \sum_{n=1}^\infty \frac{1}{n^{3+\ga }} \sin (n x_2), \  u_2 = \sg .
\]
By varying $\sg$, we get a variation direction of the initial condition,
\[
du_1(0) = 0 , \quad du_2(0) = d \sg , 
\]
which leads to the directional derivative $\pa_{\sg}F^t$ of $F^t$:
\[
\frac{\pa u_1}{\pa \sg} = \sum_{n=1}^\infty \frac{-t}{n^{2+\ga }} \cos [n (x_2-\sg t)] , \ 
\frac{\pa u_2}{\pa \sg} = 1 . 
\]
The norm of $\pa_{\sg}F^t$ is given by 
\begin{eqnarray*}
&& \| \pa_{\sg}F^t \|^2_{H^3} = 4\pi^2 + 2 \pi^2 t^2 \sum_{n=1}^\infty (1+n^2+n^4+n^6) \frac{1}{ n^{4+2\ga }} \\
&& = 4\pi^2 + 2 \pi^2 t^2 \sum_{n=1}^\infty \bigg ( \frac{1}{ n^{4+2\ga }} + \frac{1}{ n^{2+2\ga }} + \frac{1}{ n^{2\ga }} + n^{2-2\ga } \bigg ) \\
&& = \infty , \quad (t >0) , 
\end{eqnarray*}
where the last series is divergent when $\frac{1}{2} < \ga \leq 1$. Thus the directional derivative $\pa_{\sg}F^t$ does not exist. Therefore the derivative $\na F^t$ does not exist in view of the fact that the norm of the derivative $\na F^t$  is greater than or equal to the norm of the directional derivative $\pa_{\sg}F^t$. In fact, the solution operator $F^t$ is nowhere differentiable \cite{Inc12}.

The corresponding solution of (\ref{2Dsim}) to the 2D Navier-Stokes equations (\ref{2DNS}) is 
\begin{equation}
u_1 = \sum_{n=1}^\infty \frac{1}{n^{3+\ga }} e^{-\frac{n^2}{Re}t}\sin [n (x_2-\sg t)] , \  u_2 = \sg , \label{NSsim}
\end{equation}
which is also a solution in the space $C^0([0,\infty ) , H^3)$. The directional derivative $\pa_{\sg}F^t$ of $F^t$ is 
\[
\frac{\pa u_1}{\pa \sg} = \sum_{n=1}^\infty \frac{-t}{n^{2+\ga }} e^{-\frac{n^2}{Re}t}\cos [n (x_2-\sg t)] , \   \frac{\pa u_2}{\pa \sg} = 1 . 
\]
The norm of $\pa_{\sg}F^t$ is given by 
\begin{eqnarray*}
&& \| \pa_{\sg}F^t \|^2_{H^3} = 4\pi^2 + 2 \pi^2 t^2 \sum_{n=1}^\infty (1+n^2+n^4+n^6) \frac{n^2}{ n^{6+2\ga }} e^{-2\frac{n^2}{Re}t}.
\end{eqnarray*}
Let 
\begin{equation}
g(\xi ) = t^2 \xi e^{-\frac{2t}{Re}\xi } . \label{gf}
\end{equation}
The maximum of $g(\xi )$ is given by
\[
g'(\xi ) = t^2  e^{-\frac{2t}{Re}\xi } \left ( 1-\frac{2t}{Re}\xi \right ) = 0,
\]
that is,
\begin{equation}
\xi = \frac{Re}{2t} , \label{max1}
\end{equation}
where the maximal value of $g$ is 
\[
g = \frac{t \ Re}{2} e^{-1} .
\]
Thus
\[
\| \pa_{\sg}F^t \|_{H^3} \leq \bigg (\frac{1}{\sg} + \sqrt{Re} \sqrt{t} \sqrt{\frac{1}{2e}} \bigg )\| u(0) \|_{H^3} .
\]
Notice that for the full derivative $\na F^t$, the upper bound is given as \cite{Li14},
\[
\| \na F^t \|_{H^3 \ra H^3} \leq e^{\sg \sqrt{Re} \sqrt{t} +\sg_1 t } ,
\]
where $\sg$ and $\sg_1$ are two constants depending on the $H^3$ norm of the base solution. From (\ref{max1}), let 
\[
n = \bigg [ \sqrt{\frac{Re}{2t} }\bigg ] , \  \  \bigg (\text{the integer part of } \sqrt{\frac{Re}{2t}} \bigg ) ,
\]
then
\[
\frac{1}{2} \sqrt{\frac{Re}{2t}} < n \leq \sqrt{\frac{Re}{2t}}, \ \ \text{when } \sqrt{\frac{Re}{2t}} \geq 1 .
\]
We have
\begin{eqnarray*}
&& \| \pa_{\sg}F^t \|^2_{H^3}  > 4\pi^2 + \\
&& 2\pi^2 t^2 (1+n^2+n^4+n^6) \frac{n^2}{ n^{6+2\ga }} e^{-2\frac{n^2}{Re}t} \\
&& > 4\pi^2 + 2\pi^2 t^2 n^{2-2\ga }e^{-2\frac{n^2}{Re}t} \\
&& \geq 4\pi^2 + 2\pi^2 t^2 \left (\frac{1}{2} \sqrt{\frac{Re}{2t}} \right )^{2-2\ga } e^{-1} \\
&& = (2\pi )^2 +\left [ \frac{\sqrt{2}}{\sqrt{e}} \pi t^{\ga} \left (\frac{\sqrt{t} \sqrt{Re}}{2\sqrt{2}}\right )^{1-\ga } \right ]^2 \\
&& \geq \frac{1}{2}\left [ 2\pi +\frac{\sqrt{2}}{\sqrt{e}} \pi t^{\ga} \left (\frac{\sqrt{t} \sqrt{Re}}{2\sqrt{2}}\right )^{1-\ga } \right ]^2 .
\end{eqnarray*}
Thus
\[
\| \pa_{\sg}F^t \|_{H^3}  > \sqrt{2} \pi +\frac{\pi}{\sqrt{e}} t^{\ga} \left (\frac{\sqrt{t} \sqrt{Re}}{2\sqrt{2}}\right )^{1-\ga }.
\]
As $Re \ra \infty$,
\[
\| \pa_{\sg}F^t \|_{H^3}   \ra \infty , 
\]
thus
\[
\| \na F^t \|_{H^3}   \ra \infty .
\]
From another perspective, one can also obtain a lower bound for $\| \pa_{\sg}F^t \|_{H^3} $. Fix a $t >0$, for each $n$, the maximum (\ref{max1}) specifies a Reynolds number $Re^{(n)}$ (with $\xi = n^2$),
\[
Re^{(n)} = 2t n^2 ,
\]
where
\[
g(n^2) = e^{-1} t^2 n^2 . 
\]
Since 
\begin{eqnarray*}
&& \| \pa_{\sg}F^t \|^2_{H^3}  > 4\pi^2 + \\
&& 2\pi^2 t^2 (1+n^2+n^4+n^6) \frac{n^2}{ n^{6+2\ga }} e^{-2\frac{n^2}{Re}t} \\
&& > 4\pi^2 + 2\pi^2 t^2 n^{2-2\ga }e^{-2\frac{n^2}{Re}t} ,
\end{eqnarray*}
we have 
\[
\| \pa_{\sg}F^t \|_{H^3}  >2\pi \sqrt{1+\frac{t^2}{2e} n^{2-2\ga}}, \ \text{when } Re = Re^{(n)} .
\]
Therefore, as $Re^{(n)} \ra \infty$ ($n \ra \infty$), 
\[
\| \pa_{\sg}F^t \|_{H^3}  \ra \infty .
\]

More general solutions with the same property of rough dependence on initial data can be derived as follows. We expand the vorticity into a Fourier series
\[
\om = \sum_{k \in \mathbb{Z}^2/\{ 0 \}} \om_k e^{i k\cdot x}
\]
where 
\[
\om_{-k} = \overline{\om_k} , \  k = (k_1, k_2), \ x=(x_1,x_2).
\]
When the velocity has zero spatial mean, the 2D Euler equations can be re-written as
\[
\dot{\om}_k = \sum_{k=p+q} A(p,q) \om_p \om_q 
\]
where 
\begin{equation}
A(p,q) = \frac{1}{2} (|q|^{-2}-|p|^{-2})\left |\begin{array}{lr}  p_1 & q_1 \cr p_2 & q_2 \cr \end{array} \right | .
\label{Afor}
\end{equation}
When $p \parallel q$,
\[
A(p,q) = 0 .
\]
This implies that for any $k \in \mathbb{Z}^2/\{ 0 \}$,
\begin{equation}
\om = \sum_{n\in \mathbb{Z}/\{ 0 \}} C_ne^{i n (k\cdot x)} 
\label{SS}
\end{equation}
is a steady solution to the 2D Euler equation (in vorticity form) \cite{Li00}, where $C_n$'s are complex constants, and $C_{-n} = \overline{C_n}$. In terms of velocity, this solution has the form
\begin{eqnarray*}
&& u_1 = \sum_{n\in \mathbb{Z}/\{ 0 \}} \frac{ik_2C_n}{n |k|^2}e^{i n (k\cdot x)} , \\
&& u_2 = -\sum_{n\in \mathbb{Z}/\{ 0 \}} \frac{ik_1C_n}{n |k|^2}e^{i n (k\cdot x)} .
\end{eqnarray*}
Under a translation in velocity, this solution is transformed into the following time-dependent solution,
\begin{eqnarray}
&& u_1 = \sg_1+ \sum_{n\in \mathbb{Z}/\{ 0 \}} \frac{ik_2C_n}{n |k|^2}e^{i n [k_1 (x_1-\sg_1t) + k_2(x_2-\sg_2t) ]} , \non \\
&& u_2 = \sg_2 -\sum_{n\in \mathbb{Z}/\{ 0 \}} \frac{ik_1C_n}{n |k|^2} e^{i n [k_1 (x_1-\sg_1t) + k_2(x_2-\sg_2t) ]}, \label{ges}
\end{eqnarray}
where $\sg_1$ and $\sg_2$ are real constants. By choosing
\begin{equation}
C_n = \frac{1}{n^{2+\ga}} , \quad (\frac{1}{2} < \ga \leq 1) , \label{chC}
\end{equation}
we get a solution which has the same property of rough dependence on initial data as the solution (\ref{2Dsim}). The initial condition of the solution is 
\begin{eqnarray*}
&& u_1(0) = \sg_1+ \sum_{n\in \mathbb{Z}/\{ 0 \}} \frac{ik_2}{n^{3+\ga} |k|^2}e^{i n (k_1 x_1 + k_2x_2)} ,  \\
&& u_2(0) = \sg_2 -\sum_{n\in \mathbb{Z}/\{ 0 \}} \frac{ik_1}{n^{3+\ga} |k|^2}e^{i n (k_1 x_1 + k_2x_2)} .
\end{eqnarray*}
By varying $\sg_1$ (or $\sg_2$), we get a variation direction
\[
du_1(0) = d\sg_1 , \quad du_2(0) = 0 ,
\]
and the corresponding directional derivative $\pa_{\sg_1}F^t$ of the solution operator $F^t$ is 
\begin{eqnarray*}
&& \frac{\pa u_1}{\pa \sg_1} = 1+ \sum_{n\in \mathbb{Z}/\{ 0 \}} \frac{t}{n^{2+\ga}} \frac{k_1k_2}{|k|^2}e^{i n [k_1 (x_1-\sg_1t) + k_2(x_2-\sg_2t) ]} , \\
&& \frac{\pa u_2}{\pa \sg_1} = - \sum_{n\in \mathbb{Z}/\{ 0 \}} \frac{t}{n^{2+\ga}} \frac{k_1^2}{|k|^2}e^{i n [k_1 (x_1-\sg_1t) + k_2(x_2-\sg_2t) ]} .
\end{eqnarray*}
The norm of $\pa_{\sg_1}F^t$ is given by
\begin{eqnarray*}
&& \| \pa_{\sg_1}F^t\|_{H^3} = 4\pi^2 + 4\pi^2 t^2 \frac{k_1^2}{|k|^2}\sum_{n\in \mathbb{Z}/\{ 0 \}}
\bigg (  \\ &&
1+n^2|k|^2+ n^4|k|^4+ n^6|k|^6 \bigg ) \frac{1}{ n^{4+2\ga}} = \infty , 
\end{eqnarray*}
when $t > 0$ and $k_1 \neq 0$. 

Under viscous effect, the corresponding solution of (\ref{ges}) to the 2D Navier-Stokes equations (\ref{2DNS}) is given by
\begin{eqnarray}
&& u_1 = \sg_1+ \sum_{n\in \mathbb{Z}/\{ 0 \}} \frac{ik_2C_n}{n |k|^2} e^{-\frac{n^2|k|^2}{Re}t} e^{i n [k_1 (x_1-\sg_1t) + k_2(x_2-\sg_2t) ]} , \non \\
&& u_2 = \sg_2 -\sum_{n\in \mathbb{Z}/\{ 0 \}} \frac{ik_1C_n}{n |k|^2} e^{-\frac{n^2|k|^2}{Re}t} e^{i n [k_1 (x_1-\sg_1t) + k_2(x_2-\sg_2t) ]}. \label{gns}
\end{eqnarray}
We still choose $C_n$ as in (\ref{chC}). The directional derivative $\pa_{\sg_1}F^t$ is given by
\begin{eqnarray*}
&& \frac{\pa u_1}{\pa \sg_1} = 1+ \sum_{n\in \mathbb{Z}/\{ 0 \}} \frac{t}{n^{2+\ga}} \frac{k_1k_2}{|k|^2} e^{-\frac{n^2|k|^2}{Re}t} e^{i n [k_1 (x_1-\sg_1t) + k_2(x_2-\sg_2t) ]} , \\
&& \frac{\pa u_2}{\pa \sg_1} = - \sum_{n\in \mathbb{Z}/\{ 0 \}} \frac{t}{n^{2+\ga}} \frac{k_1^2}{|k|^2} e^{-\frac{n^2|k|^2}{Re}t} e^{i n [k_1 (x_1-\sg_1t) + k_2(x_2-\sg_2t) ]} .
\end{eqnarray*}
The norm of $\pa_{\sg_1}F^t$ is given by
\begin{eqnarray*}
&& \| \pa_{\sg_1}F^t\|_{H^3} = 4\pi^2 + 4\pi^2 t^2 \frac{k_1^2}{|k|^2}\sum_{n\in \mathbb{Z}/\{ 0 \}}
\bigg (  \\
&& 1+n^2|k|^2+ n^4|k|^4+ n^6|k|^6 \bigg ) \frac{n^2}{n^{6+2\ga}}e^{-2\frac{n^2|k|^2}{Re}t} .
\end{eqnarray*}
By the result on the function $g(\xi )$ (\ref{gf}) with $\xi = n^2|k|^2$, we get 
\[
 \| \pa_{\sg_1}F^t\|_{H^3} \leq \bigg ( \frac{1}{\sqrt{\sg_1^2+\sg_2^2}} +\frac{|k_1|}{|k|} \sqrt{Re} \sqrt{t} \sqrt{\frac{1}{2e}} \bigg ) \| u(0) \|_{H^3} . 
\]
From (\ref{max1}), let
\[
n = \bigg [ \frac{1}{|k|}\sqrt{\frac{Re}{2t} }\bigg ] , \  \  \bigg (\text{the integer part of } \frac{1}{|k|}\sqrt{\frac{Re}{2t}} \bigg ) ,
\]
then
\[
\frac{1}{2} \frac{1}{|k|}\sqrt{\frac{Re}{2t}} < n \leq \frac{1}{|k|}\sqrt{\frac{Re}{2t}}, \ \ \text{when } \frac{1}{|k|}\sqrt{\frac{Re}{2t}} \geq 1 .
\]
We have
\begin{eqnarray*}
&& \| \pa_{\sg_1}F^t \|^2_{H^3}  > 4\pi^2 + \\
&& 4\pi^2 t^2 \frac{k_1^2}{|k|^2}(1+n^2|k|^2+n^4|k|^4+n^6|k|^6) \frac{n^2}{ n^{6+2\ga }} e^{-2\frac{n^2|k|^2}{Re}t} \\
&& \geq 4\pi^2 + 4\pi^2 t^2 k_1^2|k|^4n^{2-2\ga }e^{-2\frac{n^2|k|^2}{Re}t} \\
&& \geq 4\pi^2 + 4\pi^2 t^2  k_1^2|k|^4 \left (\frac{1}{2} \frac{1}{|k|}\sqrt{\frac{Re}{2t}} \right )^{2-2\ga } e^{-1} \\
&& = (2\pi )^2 +\left [ \frac{2\pi}{\sqrt{e}} k_1|k|^{1+\ga} t^{\ga} \left (\frac{\sqrt{t} \sqrt{Re}}{2\sqrt{2}}\right )^{1-\ga } \right ]^2 \\
&& \geq \frac{1}{2}\left [ 2\pi +\frac{2\pi}{\sqrt{e}} k_1|k|^{1+\ga} t^{\ga} \left (\frac{\sqrt{t} \sqrt{Re}}{2\sqrt{2}}\right )^{1-\ga }\right ]^2 . 
\end{eqnarray*}
Thus
\[
\| \pa_{\sg_1}F^t \|_{H^3}  > \sqrt{2} \pi +\frac{\sqrt{2}\pi}{\sqrt{e}} k_1|k|^{1+\ga} t^{\ga} \left (\frac{\sqrt{t} \sqrt{Re}}{2\sqrt{2}}\right )^{1-\ga } .
\]
As $Re \ra \infty$,
\[
\| \pa_{\sg_1}F^t \|_{H^3}   \ra \infty , 
\]
thus
\[
\| \na F^t \|_{H^3}   \ra \infty .
\]

Now we go back to the formula (\ref{Afor}), when $|p| = |q|$,
\[
A(p,q)=0 .
\]
This implies that
\[
\om = \sum_{|k|=c} C_k e^{ik\cdot x} , \ \ ( c \text{ is a constant})
\]
is a steady solution of the 2D Euler equations (in vorticity form) \cite{Li00}, 
where $C_k$'s are complex constants, and $C_{-k} = \overline{C_k}$. In terms of velocity, this solution has the form
\begin{eqnarray*}
&& u_1 = \sum_{|k|=c} \frac{ik_2C_k}{|k|^2}e^{i k\cdot x} , \\
&& u_2 = -\sum_{|k|=c} \frac{ik_1C_k}{|k|^2}e^{i k\cdot x} .
\end{eqnarray*}
Under a translation in velocity, this solution is transformed into the following time-dependent solution,
\begin{eqnarray}
&& u_1 = \sg_1+ \sum_{|k|=c} \frac{ik_2C_k}{|k|^2}e^{i [k_1 (x_1-\sg_1t) + k_2(x_2-\sg_2t) ]} , \non \\
&& u_2 = \sg_2 -\sum_{|k|=c} \frac{ik_1C_k}{|k|^2} e^{i [k_1 (x_1-\sg_1t) + k_2(x_2-\sg_2t) ]}, \label{cks}
\end{eqnarray}
where $\sg_1$ and $\sg_2$ are real constants. A special case of (\ref{cks}) was given in \cite{HM10}. As shown in \cite{HM10}, by choosing $\sg_1$ (and/or $\sg_2$) to be $\frac{1}{n}$, $k_1$ (and/or $k_2$) to be 
$n$, and $C_k$ to be $n^{1+s}$ ($s \in \mathbb{R}$), we end up with a sequence of solutions. As $n \ra \infty$, this sequence of solutions shows that the solution operator is not uniformly continuous in $H^s$.

\section{Extension of a theorem of Cauchy to the viscous case}
Let $\xi (x,t)$ be the fluid particle trajectory starting from $x$ at the initial time $t=0$, 
\begin{equation}
\pa_t \xi (x,t) = u(\xi (x,t), t), \quad \xi (x,0) = x . \label{FPE}
\end{equation}
For any fixed $t$, one can think $\xi (x,t)$ as a map of the fluid domain. By the Liouville theorem
\[
\text{det} \left (\na_x \xi (x,t)\right ) = \text{det} \left (\na_x  \xi (x,0)\right )e^{t \na_\xi \cdot u} = 1,
\]
since 
\[
\na_\xi \cdot u = 0 , \quad \text{det} \left (\na_x  \xi (x,0)\right ) = 1 .
\]
That is, $\xi$ is a volume-preserving map in time, and any volume in the fluid domain is preserved during the fluid motion. The following is a theorem proved by Cauchy \cite{Sto80} \cite{MP94}.
\begin{theorem}
The vorticity along a fluid particle trajectory under the Euler dynamics, evolves according to 
\begin{equation}
\om_i(\xi (x,t), t) = \xi_{i,j}(x,t) \om_j(x,0). \label{CauchyT}
\end{equation}
\end{theorem}
The equation (\ref{CauchyT}) has the same form with the variation equation
\[
\dl \xi_i (x,t) = \xi_{i,j}(x,t) \dl x_j .
\]
It is not easy to extend the equation (\ref{CauchyT}) to the viscous case, while an alternative form can be extended to the viscous case. One can introduce the tensor version of vorticity
\[
\Om_{ij} = u_{i,j} - u_{j,i}.
\]
Then the theorem of Cauchy takes the following form \cite{Inc12}
\begin{theorem}
The matrix vorticity along a fluid particle trajectory under the Euler dynamics, evolves according to 
\begin{equation}
\xi_{m,i}  \Om_{mk}(\xi , t)  \xi_{k,j} = \Om_{ij}(x , 0) . \label{CauchyA}
\end{equation}
\end{theorem}
Equation (\ref{CauchyA}) is a consequence of the equation
\[
\pa_t \left (  \xi_{m,i}  \Om_{mk}(\xi , t)  \xi_{k,j} \right ) = 0 .
\]
The classical form of the Cauchy theorem (\ref{CauchyT}) can be obtained from (\ref{CauchyA}) when using the $\om$ variable. Equation (\ref{CauchyA}) plays a key role in proving the non-differentiability of the solution operator for Euler equations \cite{Inc12}. Equation (\ref{CauchyA}) can be extended to the viscous case.
\begin{theorem}
The matrix vorticity along a fluid particle trajectory under the Navier-Stokes dynamics, evolves according to 
\begin{equation}
\pa_t \left (  \xi_{m,i}  \Om_{mk}(\xi , t)  \xi_{k,j} \right ) = \frac{1}{Re} \xi_{m,i}  \Dl \Om_{mk}(\xi , t)  \xi_{k,j} .
 \label{CauchyV}
\end{equation}
\end{theorem}
\begin{proof}
Starting from the Navier-Stokes equations
\[
\pa_t u_i + u_j u_{i,j} = - p_{,i} + \frac{1}{Re} \Dl u_i ,
\]
we get
\begin{equation}
\pa_t \Om_{ij} + \Om_{ik}u_{k,j} +u_{k,i}\Om_{kj}+u_k \Om_{ij,k} = \frac{1}{Re} \Dl \Om_{ij} . 
\label{MV}
\end{equation}
From (\ref{FPE}), we have
\begin{equation}
\pa_t \xi_{i,j} = u_{i,k}\xi_{k,j} . 
\label{dxi}
\end{equation}
The material derivative of $\Om_{ij}(\xi,t)$ is given by
\begin{equation}
\pa_t \Om_{ij} (\xi,t) = \pa_t \Om_{ij}  + u_k \Om_{ij,k}  .
\label{mvd}
\end{equation}
Now we can use (\ref{dxi}) and (\ref{mvd}) to calculate the material derivative 
\begin{eqnarray*}
&& \pa_t\left (  \xi_{m,i}  \Om_{mk}(\xi , t)  \xi_{k,j} \right ) = \xi_{l,i} u_{m,l} \Om_{mk}  \xi_{k,j} \\
&& +\xi_{m,i}  [ \pa_t \Om_{mk} + u_l \Om_{mk,l} ] \xi_{k,j}  + \xi_{m,i} \Om_{mk} u_{k,l}\xi_{l,j} \\
&& = \xi_{m,i} u_{l,m} \Om_{lk}  \xi_{k,j} 
+\xi_{m,i}  [ \pa_t \Om_{mk} + u_l \Om_{mk,l} ] \xi_{k,j}  + \xi_{m,i} \Om_{ml} u_{l,k}\xi_{k,j} \\
&& = \xi_{m,i} [ \pa_t \Om_{mk} + u_l \Om_{mk,l} +\Om_{ml} u_{l,k} +u_{l,m} \Om_{lk}  ] \xi_{k,j} ,
\end{eqnarray*}
by rearranging the dummy indices. Using (\ref{MV}), we get (\ref{CauchyV}).
\end{proof}

\section{Conclusion}

By constructing explicit solutions to 2D Euler and 2D Navier-Stokes equations, we showed the rough dependence
of their solution operators upon initial data. The effect of viscosity is also studied. Finally we presented 
an extension of a theorem of Cauchy to the viscous case.

\end{document}